\documentclass[11pt, oneside, reqno]{amsart}

\usepackage[top=2.5cm, bottom=2.2cm,left=2cm,right=2cm,footnotesep=20pt]{geometry}
\usepackage{amsmath, amsthm, amssymb, mathrsfs}
\usepackage[all, cmtip]{xy}
\usepackage[titletoc,title]{appendix}
\usepackage[usenames,dvipsnames]{color}
\usepackage{slashed} 
\usepackage{xfrac}
\usepackage{caption}
\usepackage{subcaption}
\usepackage{eucal}
\usepackage{verbatim}
\usepackage[backend=bibtex,
            isbn=false,
            doi=false,
            maxbibnames=5,
            giveninits=true,
            style=alphabetic,
            citestyle=alphabetic]{biblatex}
\renewbibmacro{in:}{}
\usepackage{enumerate}
\usepackage{multirow}
\usepackage{array}
\usepackage{tikz}
\usetikzlibrary{decorations.pathmorphing,decorations.markings}
\usepackage[bookmarks=true, linkbordercolor={0 0 1}]{hyperref}
\usepackage{epstopdf}
\theoremstyle{definition} \newtheorem{definition}{Definition}[section]

\newtheorem{lemma}[definition]{Lemma}
\newtheorem{theorem}[definition]{Theorem}
\newtheorem{prop}[definition]{Proposition}
\newtheorem{conjecture}[definition]{Conjecture}
\newtheorem{corollary}[definition]{Corollary}

\newtheorem*{lemma*}{Lemma}

\newtheorem{thm}{Theorem}[section]

\theoremstyle{definition}
\newtheorem{dfn}[thm]{Definition}
\newtheorem{dfn/lem}{Definition/Lemma}

\theoremstyle{remark}

\theoremstyle{definition} \newtheorem{remark}[definition]{Remark}
\theoremstyle{definition} 
\theoremstyle{definition} 
\theoremstyle{definition} 
\theoremstyle{definition} 
\theoremstyle{definition} 

\def\clie{{\rm C}_{\rm Lie}}

\def\cloc{{\rm C}_{\rm loc}}
\def\oloc{\mathcal{O}_{\rm loc}}
\def\bu{\bullet}
\def\Bar{\overline}
\def\Hat{\widehat}

\def\xto{\xrightarrow}
\def\Sym{{\rm Sym}}

\def\<{\langle}
\def\>{\rangle}

\def\cE{\mathcal E}
\def\cL{\mathcal L}
\def\cM{\mathcal M}\def\cN{\mathcal N}\def\cO{\mathcal O}\def\cP{\mathcal P}
\def\cQ{\mathcal Q}\def\cS{\mathcal S}\def\cT{\mathcal T}

\def\CC{\mathbb C}

\def\RR{\mathbb R}

\def\ZZ{\mathbb Z}

\def\fg{\mathfrak g}\def\fh{\mathfrak h}

\def\fw{\mathfrak w}

\newcommand{\dd}{\partial}
\newcommand{\dbar}{\overline{\partial}}

\renewcommand{\gg}{\mathfrak{g}}

\renewcommand{\sp}{\mathfrak{sp}}

\newcommand{\GL}{\mathrm{GL}}

\newcommand{\Sp}{\mathrm{Sp}}

\newcommand{\U}{\mathrm{U}}
\newcommand{\spin}{\mathrm{Spin}}

\newcommand{\iso}{\cong}

\newcommand{\sub}{\subseteq}

\newcommand{\inj}{\hookrightarrow}

\newcommand{\bb}[1]{\mathbb{#1}}

\newcommand{\mr}[1]{\mathrm{#1}}
\newcommand{\mc}[1]{\mathcal{#1}}
\newcommand{\mf}[1]{\mathfrak{#1}}
\newcommand{\ol}[1]{\overline{#1}}
\newcommand{\ul}[1]{\underline{#1}}

\renewcommand{\hom}{\operatorname{Hom}}

\DeclareMathOperator{\dens}{Dens}

\def\d{{\rm d}}
\newcommand{\Obs}{\mathrm{Obs}}

\newcommand{\obs}{\mathrm{Obs}}

\newcommand{\dR}{\mathrm{dR}}

\renewcommand{\U}{\mathrm{U}}
\newcommand{\ham}{/\!\!/}

\let\OLDthebibliography\thebibliography
\renewcommand\thebibliography[1]{
  \OLDthebibliography{#1}
  \setlength{\parskip}{0pt}
  \setlength{\itemsep}{4pt plus 0.3ex}
}

\def\ham{\fh}
\def\define{\overset{\rm def}{=}}
\def\ep{{\varepsilon}}

\def\cred{{\rm C}_{\rm red}}

\bibliography{Twist}

\title{Holomorphic Poisson Field Theories}

\author{Chris Elliott}
\address{University of Massachusetts, Amherst \\ Department of Mathematics and Statistics \\ 710 N Pleasant St \\ Amherst, MA 01003}
\email{celliott@math.umass.edu}

\author{Brian R. Williams}
\address{School of Mathematics\\
University of Edinburgh \\ 
Edinburgh \\ 
UK}
\email{brian.williams@ed.ac.uk}

\date{\today}

\begin{document}

\begin{abstract}
We construct a class of quantum field theories depending on the data of a holomorphic Poisson structure on a piece of the underlying spacetime. 
The main technical tool relies on a characterization of deformations and anomalies of such theories in terms of the Gelfand--Fuchs cohomology of formal Hamiltonian vector fields.
In the case that the Poisson structure is non-degenerate such theories are topological in a certain weak sense, which we refer to as ``de Rham topological". 
While the Lie algebra of translations acts in a homotopically trivial way, we will show that the space of observables of such a theory {\em does not} define an $E_n$-algebra. 
Additionally, we will highlight a conjectural relationship to theories of supergravity in four and five dimensions. 
\end{abstract}

\maketitle

\section{Introduction}

This paper is devoted to the construction of a novel family of quantum field theories defined on the product of a holomorphic Poisson manifold $X$ and a smooth manifold $M$.  The action functional of these theories is analogous to the BF theory action functional, but where the interaction term is defined using the holomorphic Poisson bracket on $X$, so we refer to these theories as holomorphic \emph{Poisson BF theories}.  We will show that these theories have a number of interesting properties.
\begin{enumerate}
 \item The only potential anomaly to quantization of Poisson BF theory occurs at first order in $\hbar$.  
 On $\CC^{2n} \times \RR^m$, this anomaly is given by a class in the cohomology of the infinite-dimensional Lie algebra $\ham_{2n}$ of Hamiltonian vector fields on $\CC^{2n}$.  We show that if $m \le 6$ this cohomology class vanishes, and so the Poisson BF theory can be quantized to all orders (Section \ref{sec:quantization}).
 \item Poisson BF theory is gravitational in nature. 
 Indeed, there is a fundamental field of Poisson BF theory of spin $2$ and the interactions contain spacetime derivatives.
 Moreover, if $X$ is a holomorphic symplectic surface, Poisson BF theories appear to arise as twists of minimal supergravity theories in dimensions 4 and 5.  Such supergravity theories occur by dimensionally reducing 11-dimensional $\mc N=1$ supergravity on a $G2$-manifold or Calabi--Yau 3-fold respectively (Section \ref{sugra_section}).
 \item Poisson BF theories are examples of theories that are ``de Rham topological'' but not ``Betti topological''.  That is, their algebras of observables are not locally constant, but the group of translations acts homotopically trivially.  The factorization algebra of observables has the structure of an algebra over the operad $C_\bullet(\mr{Disk}_{2n+m}^{\mr{col}})$ of singular chains for colored little $(2n+m)$-disks (Section \ref{factn_section}).
\end{enumerate}

\subsection{Topological Field Theories}
In the literature on topological quantum field theories, a distinction is often drawn between topological theories of \emph{Schwarz} type, and of \emph{Witten} (or cohomological) type \cite{SchwarzICM}.  A field theory is topological of Schwarz type if the action functional is strictly independent of the metric, so the gravitational stress-energy tensor $\mc T^{\mu \nu} = \frac{\delta S}{\delta g^{\mu \nu}}$ vanishes.  Being topological of Witten-type is slightly more subtle, one usually requires the existence of an odd observable $Q$ satisfying the condition $[Q,Q] = 0$, such that $\mc T^{\mu \nu} = [Q, \mc G^{\mu \nu}]$ for some family of observables $\mc G^{\mu \nu}$.

This condition can be phrased more naturally in the Batalin--Vilkovisky (BV) formalism, which provides for a description of classical and quantum field theory in terms of homological algebra and Lie theory, or equivalently in terms of formal derived geometry.  In the BV formalism, the local observables form a cochain complex that is further equipped with a bracket of cohomological degree $+1$.  
We refer to the differential on this cochain complex as the ``classical BV complex" and the bracket as the ``BV bracket". 
Now, a theory is topological of Schwarz type if the stress-energy tensor strictly vanishes, and it is topological of Witten type if the stress-energy tensor vanishes \emph{up to homotopy}, i.e. if it is an exact element of the cochain complex.  From a homotopy-theoretic point of view only the latter notion is invariant: it is often possible to describe a perturbative quantum field theory in many different homotopy-equivalent ways, and a theory may be topological of Schwarz type in some models, but only topological of Witten type in others.

Witten-type topological field theories appear naturally when one studies twisting for supersymmetric field theories.  Given a supersymmetric field theory on $\RR^{p,q}$ and an odd element $Q$ of the supersymmetry algebra such that $[Q,Q]=0$, one can \emph{twist} the supersymmetric theory by adding the action of $Q$ to the BV differential on the algebra of local observables.  According to our discussion above, the resulting twist will be \emph{topological} if the following condition holds:
\begin{itemize}
 \item[T1:] $\mc T^{\mu \nu} = [Q, \mc G^{\mu \nu}]$ for some $\mc G^{\mu \nu}$.
\end{itemize}
On the other hand, there is a weaker condition that is often checked first, that only depends on $Q$ itself and not on the choice of theory.  
The super Lie algebra of supertranslations is an extension of the abelian Lie algebra $\RR^{p,q}$ of translations by an odd space $\Pi \Sigma$, where $\Sigma$ is a spinorial representation of $\spin(p,q)$ (so $Q$ is an element of $\Sigma$).
\begin{itemize}
 \item[T2:] The map $[Q, -] \colon \Sigma \to \RR^{p,q}$ is surjective.
\end{itemize}
This condition T2 is guaranteed, for example, if we choose a twisting homomorphism from $\spin(p,q)$ to the group $G_R$ of R-symmetries such that $Q$ is invariant under the twisted $\spin(p,q)$-action (this has the additional advantage of allowing the twisted theory to be defined on all oriented manifolds).  

It is possible to rephrase condition T2 in language similar to condition T1.  The \emph{canonical} stress-energy tensor is the conserved current associated to the translation symmetry, condition T2 says that the canonical stress-energy tensor is $Q$-exact.  The canonical stress-energy tensor and the gravitational stress-energy tensor $T^{\mu\nu}$ do not typically coincide, they differ by a correction term (the Belinfante--Rosenfeld correction) derived from the conserved currents for the rotation action.

Condition T1 implies T2, but are there examples satisfying only T2?  Let us discuss what the two conditions tell us, formally, about the algebra of local observables.

\subsection{Betti vs. de Rham topological factorization algebras}
The theory of factorization algebras provides an effective mathematical model for the study of classical and quantum field theories, describing their dependence on an open subset of spacetime, and their ``operator products".  
It is the central premise of Costello and Gwilliam's work in \cite{Book1,Book2} that one can associate to every quantum field theory on a manifold $M$ a factorization algebra of ``observables" on $M$. 
Let us provide an informal definition.

\begin{definition} \label{dfn:fact}
A \emph{prefactorization algebra} on a smooth manifold $M$ is an assignment of a cochain complex $\obs(U)$ to each open subset $U \sub M$, together with structure maps $\obs(U_1) \otimes \cdots \otimes \obs(U_n) \to \obs(V)$ for each set of pairwise disjoint subsets $U_1, \ldots, U_n \sub V$, satisfying natural compatibility conditions.  A \emph{factorization algebra} is a prefactorization algebra satisfying a descent condition (see \cite[Chapter 6]{Book1} for a precise definition).
\end{definition}

In general, the open set $\obs(U)$ is highly dependent on the geometry of $U$, for instance if $U \sub V$ are a pair of homotopy equivalent open subsets, the map $\obs(U) \to \obs(V)$ will typically be far from a quasi-isomorphism.  However, if $\obs$ describes the local observables in a topological quantum field theory, this should not be the case: the local observables should be topological in the following sense.

\begin{definition}
A factorization algebra $\obs$ is \emph{Betti topological} if, whenever $U \inj V$ is a homotopy equivalence, the associated map $\obs(U) \to \obs(V)$ is a quasi-isomorphism.
\end{definition}

\begin{remark}
A Betti topological factorization algebra is equivalent to what is often referred to as a locally constant factorization algebra.
\end{remark}

For example, the local observables in a Schwarz-type topological quantum field theory will manifestly form a Betti topological factorization algebra, because the fields and the action functional will not be sensitive to the local geometry of an open set.  The same is true of a Witten-type topological quantum field theory, such as a topological twist satisfying the condition T1 (see \cite[Section 3.5]{ElliottSafronov}).  When we consider twists only satisfying the condition T2 however, we can only guarantee an apparently weaker.  Such theories are \emph{de Rham} topological, in the following sense.

\begin{definition}
A factorization algebra $\obs$ on $\RR^{n}$ is \emph{de Rham topological} if it is translation invariant, meaning that there is a smooth action of the group $\RR^{n}$, and the translation action is \emph{homotopically trivial}, meaning that there is additionally an infinitesimal action of the abelian Lie algebra $\RR^{n}$ of cohomological degree $-1$, say
\[\eta \colon \RR^{n} \to \mr{Der}(\obs),\]
such that $\d \eta(v) = \dd v$ for all $v \in \RR^{n}$.
\end{definition}

In \cite{ElliottSafronov}, Safronov and the first author studied these de Rham topological factorization algebras, and learned that they are not in fact that distant from Betti topological theories.  To explain their relationship, we will introduce a concept from homotopical algebra, the notion of an $\bb E_n$-algebra.  The $\bb E_n$-operad is an operad valued in cochain complexes.  It is built by taking the singular chain complex of the operad $\mr{Disk}_n$ of little $n$-disks: the operad valued in smooth manifolds whose space of order $k$ operations is the space of embeddings of $k$ $n$-disks into the unit $n$-disk.  The relationship between $\bb E_n$-algebras and factorization operators begins with a famous result of Lurie.

\begin{theorem}[{\cite[Theorem 5.4.5.9]{HA}}]
There is a fully faithful embedding from the $\infty$-category of $\bb E_n$-algebras into the $\infty$-category of factorization algebras, whose essential image consists of Betti topological factorization algebras.
\end{theorem}

One way of stating this result is to say that a Betti topological factorization algebra is completely determined by its value on the unit disk.  When we consider de Rham topological factorization algebras, it turns out that they are determined by their values on disks, but but with dependence on scale.  There is a colored operad $\mr{Disk}_n^{\mr{col}}$, with space $\RR_{>0}$ of colors corresponding to radii: this is the colored operad valued in smooth manifolds whose space of order $k$ operations $\mr{Disk}_n^{\mr{col}}(r_1, \ldots, r_k | R)$ is the space of embeddings of $k$ $n$-disks of radius $r_1, \ldots, r_k$ into a disk of radius $R$.

\begin{theorem}[{\cite[Theorem 2.23 and 2.29]{ElliottSafronov}}]
There is an equivalence of $\infty$-categories between the $C_\bullet(\mr{Disk}_n^{\mr{col}})$-algebras and de Rham topological factorization algebras.  The essential image of the $\infty$-category of $\bb E_n$ algebras consists of those de Rham topological factorization algebras where the map $\obs(B_r(0)) \to \obs(B_R(0))$ is a quasi-isomorphism, where $B_r(0) \sub B_R(0)$ are concentric $n$-disks about 0 with radii $r < R$.
\end{theorem}

We can sloganize this result by saying that the difference between theories satisfying condition T1 and condition T2 is the presence or absence of \emph{dilation invariance}.

\subsection{Goals of this Paper}
The main aim of this paper is to describe a class of quantum field theories whose local observables are de Rham but not Betti topological.  The examples we construct will be defined on the product of a holomorphic Poisson manifold and an oriented real manifold, and theories of this type in dimension 4 and 5 conjecturally occur as twists of minimal supergravity theories, as we will discuss in Section \ref{sugra_section}.

The basic idea is as follows.  There is a procedure in the BV formalism for constructing classical field theories on $M$ of ``cotangent type''.  We start with a sheaf of differential graded Lie algebras $(\mc L, \d)$ on $M$.  The cotangent theory associated to $L$ has fields given by elements $(A,B)$ of the cochain complex $\mc L \oplus \mc L^![-3]$, where $L^! = \hom(L, \dens_M)$, with an action functional of the form
\[S(A,B) = \int_M \langle B, \d A + \frac 12 [A,A] \rangle.\]
The most famous example of a field theory of this form is BF theory, which is generated by this procedure from the dg Lie algebra $L = \Omega^\bullet(M, \gg)$ for a Lie algebra $\gg$.

We will define a classical field theory of cotangent type on $\CC^{2n} \times \RR^m$ starting with the cochain complex $L = \Omega^{0,\bu}(\CC^{2n}) \, \Hat{\otimes} \, \Omega^\bu(\RR^m)$.  
We then equip this complex with a Lie bracket coming from the standard holomorphic symplectic structure on $\CC^{2n} \iso T^*\CC^n$.  
Concretely, this Lie bracket is defined on the $\Omega^{0,\bu}(\CC^{2n})$ factor of $L$ by the formula
\[[f \d \ol z^{i_1} \wedge \cdots \wedge \d \ol z^{i_k}, g \d \ol z^{j_1} \wedge \cdots \wedge \d \ol z^{j_\ell}] = \{f,g\} \d \ol z^{i_1} \wedge \cdots \wedge \d \ol z^{i_k} \wedge \d \ol z^{j_1} \wedge \cdots \wedge \d \ol z^{j_\ell}.\]
We will refer to this theory as an example of \emph{Poisson BF theory}, because it has an action functional analogous to that of BF theory, with bracket associated to the holomorphic Poisson structure on $\CC^{2n}$.  

\begin{remark}
We will define theories of this type much more generally: on the product $X \times M$ of a holomorphic Poisson manifold and a real oriented manifold.  We emphasise that this construction is not purely arbitrary from a physical point of view.  We explain in Section \ref{sugra_section} that in low dimensions, Poisson BF theories appear to arise as partially topological twists of supergravity theories.
\end{remark}

This classical field theory admits a natural quantization which is manifestly translation invariant. 
We argue in Section \ref{Quantization_section} that classical Poisson BF theories can, for purely formal reasons, be extended to 1-loop exact quantum theories as long as we verify that a certain 1-loop anomaly vanishes.  We show that this occurs for all $n$, and all $m \le 6$, by relating the anomaly to a certain cocycle in the cohomology of an infinite-dimensional Lie algebra, which we demonstrate to be a coboundary in Theorem \ref{thm:quantization}. 

Poisson BF theory is not Betti topological: we can see this straight away by constructing its algebra of observables and checking that the inclusion of concentric balls is not an equivalence.  On the other hand, we can see that it is automatically quite close to being de Rham topological, without us needing to do much work.  The group $\RR^{2n} \times \RR^m$ of translations in the antiholomorphic $\ol z^1, \ldots, \ol z^{2n}$ and real $t^1, \ldots, t^m$ directions acts homotopically trivially because the action functional only depends on a complex structure on $\CC^{2n}$ (for an explicit potential, see Section \ref{translation_section}).  Less obviously, we will also show the following.

\begin{theorem}[{See Theorem \ref{de_rham_holo_action_thm}, Theorem \ref{quantum_lc_thm}}]
The Lie algebra of holomorphic translations on $\CC^{2n}$ acts on Poisson BF theory in a homotopically trivial way. 
Therefore the factorization algebra of classical observables of Poisson BF theory is de Rham topological.  When a quantization exists, the factorization algebra of quantum observables is also de Rham topological.
\end{theorem}

\begin{remark}
The Lie algebra of holomorphic translations spanned by $\{\partial_{z_1}, \ldots, \partial_{z_{2n}}\}$ is an abelian sub Lie algebra of holomorphic functions modulo constant functions $\cO^{\rm hol}(\CC^{2n}) / \CC$. 
In fact, we will see that there is a homotopically trivial action of the dg Lie algebra $\Omega^{0,\bu}(\CC^{2n})$, which resolves the Lie algebra of holomorphic functions, on Poisson BF theory.
\end{remark}

\begin{remark}
We should point out that our example of a theory with is de Rham but not Betti topological does not provide an example of a \emph{twist} where the supercharge $Q$ satisfies condition T2 but not T1.  Although we conjecture that Poisson BF theory does arise as a twist, the corresponding supercharge will only be holomorphic-topological, i.e. it will only satisfy T2 for \emph{some} directions.  In fact, for supersymmetric Yang--Mills theories we know that all twists satisfying T2 also satisfy T1!  This follows by exhaustive classification of all possible twists \cite{ESW}.  Additionally, twists of superconformal theories satisfying condition T2 also satisfy T1, as shown in \cite{ElliottSafronov}.  We speculate that in fact all twists satisfying condition T2 (making them a priori only de Rham topological) also satisfy condition T1 (making them Betti topological, or Witten type), but at the moment we don't know a general argument justifying this in all cases. 
\end{remark}

\section*{Acknowledgements}
Our main example is based on a suggestion made by Kevin Costello during a visit to the Perimeter Institute in 2017.  
Research at Perimeter Institute is supported by the Government of Canada through Industry Canada and by the Province of Ontario through the Ministry of Economic Development \& Innovation.
B.R.W. is partially supported by the National Science Foundation Award DMS-1645877 and by the National Science Foundation under Grant No. 1440140, while the author was in residence at the Mathematical Sciences Research Institute in Berkeley, California, during the semester of Spring 2020.

\section{Classical Field Theories on Holomorphic Poisson Manifolds}
Let $X$ be a holomorphic Poisson manifold of dimension $d$ with holomorphic Poisson tensor $\Pi \in \wedge^2 T^{1,0}_X$.  
The Poisson tensor equips the sheaf of holomorphic functions $\cO^{\mr{hol}}(X)$ with a Lie bracket that we denote by $\{-,-\}_\Pi$. 
This bracket extends to a bracket on the Dolbeault complex $\Omega^{0,\bu}(X)$.  Recall that a \emph{local Lie algebra} on $X$ is a sheaf of differential graded Lie algebras on $X$ where the differential and the bracket are given by differential operators (\cite[Section 6.2]{Book1}).

\begin{definition}\label{dfn:localLie1}
Let $\mc L_{\Pi}$ be the local Lie algebra on $X$ with underlying cochain complex $(\Omega^{0,\bullet}(X), \ol \dd)$, and with Lie bracket given by the formula 
\[[\alpha, \beta] = (-1)^{|\alpha||\beta| +1} \{\alpha \wedge \beta\}_\Pi .\]
That is, the unique graded antisymmetric extension of the Poisson bracket of functions on $X$ to the Dolbeault complex. 
\end{definition}

This is a local Lie algebra defined on any holomorphic Poisson manifold. There is a related local Lie algebra defined on product manifolds of the form $X \times M$, where $M$ is a smooth (not necessarily complex) manifold. This local Lie algebra will be manifestly topological (of Schwarz type) along $M$.

\begin{definition}
Let $Y = X \times M$, where $M$ is a smooth manifold and $X$ is a holomorphic Poisson manifold as above.  
Let $\mc L_{\Pi,M}$ be the local dg Lie algebra on $Y$ whose sheaf of sections is
\[
\mc L_\Pi \; \Hat{\otimes} \; \Omega^\bu(M) = \Omega^{0,\bu}(X) \; \Hat{\otimes}\; \Omega^\bu(M).
\]
The dg Lie algebra structure is obtained by tensoring the dg Lie algebra structure on $\mc L_\Pi$ with the commutative dg algebra $\Omega^\bu(M)$.  
\end{definition}

As a trivial remark, note that in the case where $M$ is a point, thought of as a zero dimensional manifold, the local Lie algebra $\mc L_{\Pi,M}$ becomes the original local Lie algebra $\cL_\Pi$. 

Using $\mc L_{\Pi,M}$ we arrive at the definition of the classical field that we will focus on for the remainder of the paper.  Let us first state the definition of a classical field theory in the Batalin--Vilkovisky (BV) formalism \cite{BatalinVilkovisky} that we will be studying.

\begin{definition}
A \emph{classical field theory} on a manifold $X$ is a sheaf $\mc E$ of cochain complexes on $X$, where the shift $\mc E[1]$ is equipped with the structure of a local Lie (or more generally, $L_\infty$) algebra, and $\mc E$ is equipped with a $(-1)$-shifted symplectic pairing $\mc E \otimes \mc E[-1] \to \dens_X$.  The \emph{cotangent theory} to a local Lie algebra $\mc L$ is the direct sum $\mc L[1] \oplus \mc L^![-2]$, where $\mc L^! = \mr{Hom}(\mc L, \dens_X)$, equipped with its canonical shifted symplectic structure.
\end{definition}

\begin{remark}
For a discussion of where this definition comes from, and what it means, we refer the reader to the extensive discussion in \cite{CostelloBook, Book2} and \cite[Section 1.1]{ESW}.
\end{remark}

The focus of this paper is the following example.
\begin{definition}
We define \emph{Poisson BF Theory} on $Y = X \times M$ to be the cotangent theory to the local Lie algebra $\mc L_{\Pi, M}$.  
We will denote this theory by $\mc E_{\Pi}$.
\end{definition}

\begin{remark}
Such theories exist in slightly more generality.
Indeed, we can consider holomorphic Poisson BF theory on the total space $Y$ of any smooth fibration whose base is equipped with a holomorphic Poisson structure. 
Since we will be only interested in the flat case in what follows, we will not return to this general situation. 
\end{remark}

Equivalently, a field theory in the BV formalism is specified by a space of fields equipped with a $(-1)$-shifted symplectic structure together with a {\em local action functional} $S$ on the fields which satisfies the classical master equation $\{S,S\}_{\mr{BV}} = 0$.
Here, the bracket $\{-,-\}_{\mr{BV}}$ denotes the BV bracket induced from the shifted symplectic structure.

Spelling this definition out, the fields consist of pairs 
\[
(A, B) \in \Omega^{0,\bullet}(X) \; \Hat{\otimes}\; \Omega^\bu(M) [1] \; \oplus \; \Omega^{d,\bullet}(X) \; \Hat{\otimes} \;\Omega^\bu(M) [d+m-2] 
\] 
where $d = \dim_{\CC}(X)$ and $m = \dim_{\RR}(M)$. 
The local functional representing the BV action is
\begin{equation}\label{eqn:action2}
S_{\Pi,M}(A,B) = \int_{Y} B \wedge \left( \ol \dd A + \frac 12 \{A \wedge A\}_\Pi\right).
\end{equation}
It will be useful in the next section to split the action functional above as $S = S_{\rm free} + I$ where $I = \frac 12 \int_Y  B \wedge \{A \wedge A\}_\Pi$. 

\begin{remark}
A special case of the above construction is when the holomorphic Poisson structure $\Pi$ is nondegenerate, so that $X$ is a holomorphic symplectic manifold. 
In this case, we will find that the theories we have introduced locally possess de Rham translation invariance.
This is not true in the case that $\Pi$ is genuinely degenerate.
\end{remark}

\begin{remark}
There is a related theory studied by Costello in \cite{CostelloM2}.  One can define a theory on $\CC^{2n} \times \RR^m$ with the same classical BV complex of fields as we have discussed here, but where the interaction is defined not using the Poisson bracket, but instead using the a quantization thereof: the \emph{Moyal star-product}.  Costello studies theories of this type where $n=m=1$, and argues that such theories arise from supergravity in a suitable twisted $\Omega$-background.
\end{remark}

\section{Quantization and Anomalies} \label{sec:quantization}

From now on we will mostly restrict attention to the cotangent theory associated to the local Lie algebra $\mc L_{\Pi,\RR^m}$ on $Y = \CC^{2n} \times \RR^m$ where $X = \CC^{2n}$ is equipped with its standard holomorphic symplectic structure.  
We will denote the local Lie algebra $\mc L_{\Pi,\RR^m}$ where $\Pi$ is the standard Poisson structure on $\CC^{2n}$ simply by $\cL_{n,m}$.  
In this section we will use the BV formalism as developed in \cite{CostelloBook, Book2} together with special results for mixed holomorphic-topological theories as in \cite{BWhol, GWcs, GWR}.

\subsection{Background on Quantization} \label{sec:renorm}

For now, let us fix the data of a classical field theory in the BV formalism. 
As we've discussed, this consists of a $(-1)$-symplectic sheaf of fields $\cE$ together with a dg Lie structure.  This can equivalently be encoded in terms of a local action functional $S_0 = S_{\rm free} + I$ satisfying the classical master equation
\[
\{S, S\}_{\mr{BV}} = 0, \text{ or equivalently } Q I + \frac 12 \{I, I\}_{\mr{BV}} = 0 .
\]
Here, $I$ is an element of the space $\oloc(\cE)$ of {\em local functionals} on $\cE$. 
This space consists of functionals on the fields which are given by Lagrangian densities, where we quotient out by total derivatives.
For a precise definition we refer to \cite[Definition 3.5.1.1]{Book2}. 

The starting point in the renormalization group approach to quantum field theory is an effective family of $\hbar$-dependent functionals $\{I[L]\}$ on the space of fields parametrized by a ``length scale" $L > 0$. Heuristically, the ``$L \to 0$" limit, although na\"ively ill-defined, represents the full quantum action of the field theory. 

In order to make sense of the quantum action, for each $L > 0$ the functional $I[L] \in \cO(\cE)[[\hbar]]$ must satisfy various conditions, which can be found in \cite[Definition 8.2.9.1]{Book2}. 
The first condition this family must satisfy is that its $\hbar \to 0$ and $L \to 0$ limit agrees with the classical interaction $I$ defining the field theory. In addition, the family of functionals must satisfy  (1) the {\em renormailzation group equation} and (2) the {\em quantum master equation}. 

The renormalization group equation says that
\[
I[L'] = W_{L<L'} (I[L])
\]
where $W_{L<L'}$ is an isomorphism called the renormalization group flow. 
It is defined as a sum over weights of graphs: the Feynman diagram expansion. 
Very roughly, the weight is obtained by placing the classical interaction at each vertex of the graph and labelling the edges by elements of $\mc E^{\otimes 2}$ called ``propagators", then contracting tensors according to the shape of the diagram. 
A family $\{I[L]\}$ of interactions satisfying the renormalization group equation is called an {\em effective family}, or {\em prequantization}. 
For more on this construction see Chapter 2 of \cite{CostelloBook}.

One of the main results of \cite{CostelloBook} is that an effective family exists for any classical field theory. In general, however, to construct the family involves some serious analysis involving the introduction of ``counterterms". 
Thankfully, for the class of theories we consider, the effective families can be understood completely explicitly, without the introduction of counterterms. 

The first step is the following easy combinatorial observation.  One can check that any effective family for a cotangent field theory necessarily only involves graphs that have genus at most one. Thus, any effective family is at most linear in its $\hbar$ expansion, referred to as ``one-loop exact" effective family.

The second step is a consequence of the following result, which states that for theories of mixed holomorphic-topological type there exists a one-loop quantization which is void of counterterms.  A classical field theory on $\CC^n \times \RR^m$ is called a theory of mixed holomorphic-topological type if there is a homotopically trivial action (see Definition \ref{dR_action_def}) of the group $\RR^n \times \RR^m$, where $\RR^n$ acts on $\CC^n$ by translations in the antiholomorphic directions.

\begin{theorem}[\cite{BWhol, GWcs, GWR}] \label{thm:oneloop}
For any mixed holomorphic-topological theory on $\CC^n \times \RR^m$, $n,m \geq 0$, there exists 
a translation-invariant effective family $\{I[L]\}$ which to first order in $\hbar$ is finite. 
\end{theorem}

\begin{remark}
This theorem only concerns the effective family of holomorphic-topological theories to first order in $\hbar$. 
It does not imply that counterterms of order $\hbar^n$, $n > 1$ vanish. 
\end{remark} 

Since any theory of BF type is exact at one-loop, this theorem implies that a translation invariant effective family for Poisson BF theory exists on $\CC^{2n} \times \RR^m$ for any $n,m$ and that all counterterms vanish.

Once an effective family is constructed, the next condition required of a quantization in the BV formalism is the {\em quantum master equation} (henceforth abbreviated QME). Heuristically, if $I^{\mr{q}}$ denotes the na\"ive quantum action, the QME reads
\[(Q + \hbar \triangle) e^{I^{\mr{q}}/\hbar} = 0.\]
Where $\triangle$ is the ``scale zero" BV Laplacian associated to the shifted symplectic structure defining the classical BV theory. 
There are two problems with this equation: first, the ``scale zero" BV Laplacian is ill-defined as it involves contractions of distributions; second, $I^q$ is only defined by means of an effective family $\{I[L]\}$ as described above. 

To make sense of this, one introduces a regularized QME at each scale $L$, which can equivalently be written as
\[
Q I[L] + \hbar \triangle_L I[L] + \frac{1}{2} \{I[L], I[L]\}_L = 0 .
\]
Here $\triangle_L$ is the well-defined regularized scale $L$ BV Laplacian, and $\{-,-\}_L$ is a regularized version of the classical BV bracket.
The $\hbar \to 0$, $L \to 0$ limit of the above equation is precisely the classical master equation. 
An effective family $\{I[L]\}$ is a quantum field theory if it satisfies the scale $L$ QME for every $L > 0$. 
For more complete details we refer to \cite[Chapter 8]{Book2}. 

In general, not every effective theory satisfies the QME. 
The {\em scale} $L$ {\em obstruction} to satisfying the QME describes the failure of $I[L]$ to satisfy the scale $L$ QME.
Since $I[L]$ is filtered by powers of $\hbar$, so is the obstruction.
For theories of cotangent type as considered in this paper, $I[L]$ truncates at order $\hbar$, hence we only need to consider the $\hbar$-linear obstruction which we denote by $\hbar \Theta[L]$. 
The $L \to 0$ limit of $\Theta[L]$ is defined and determines a cohomological degree $+1$ {\em local functional}
\[
\Theta \define \lim_{L \to 0} \Theta [L] \in \oloc(\cE) .
\]
Moreover, $\Theta$ is closed for the classical differential $\{S,-\}_{\mr{BV}}$, hence determines a cohomology class $[\Theta] \in \mr H^1(\oloc(\cE), \{S,-\}_{\mr{BV}})$, see \cite[\S 5.11]{CostelloBook}.

In fact, for theories of cotangent type such as Poisson BF theory, we have more control over the obstruction. 
Before stating this result, we recall that the space of fields of our theory is of the form $\cE = \cL_{n,m} [1] \oplus \cL_{n,m}^! [-2]$. 
On can identify the operator $\{S, -\}_{\mr{BV}}$ acting on local functionals $\oloc(\cE)$ with the Chevalley--Eilenberg differential for the dg Lie algebra $\cE[-1] = \cL_{n,m} \ltimes \cL_{n,m}^![-3]$. 
Therefore, we will use the notation $\cloc^\bu(\cE[-1])$, which we refer to as the local Chevalley--Eilenberg cocahin complex, for the cochain complex $(\oloc(\cE), \{S,-\}_{\mr{BV}})$. 
Notice further that there is an embedding of cochain complexes $\cloc^\bu(\cL_{n,m}) \hookrightarrow \cloc^\bu(\cE[-1])$ consisting of local cochains which depend only on the $A$-fields.

\begin{lemma}\label{lemma:Atype}
Let $\hbar \Theta_{n,m}$ be the one-loop obstruction for Poisson BF theory on $\CC^{2n} \times \RR^m$ satisfying the QME. 
Then $\Theta_{n,m}$ is a degree $+1$ cocycle of the local cohomology $\cloc^\bu(\cL_{n,m})$. 
\end{lemma}

\begin{proof}
We use a general result of Costello \cite[Corollary 16.0.5]{CostelloWittengenus} which states that the one-loop anomaly of any BV theory reduces to the weight of wheel graphs.  From here, the proof is purely combinatorial.
The classical interaction for a theory of BF type is generally of the form $\int B [A, A]$.
Thus, for a wheel graph whose vertices are labeled by this interaction, the weight is purely a function of the $A$-field and the result follows.
\end{proof}

\subsection{Quantization of Poisson BF Theory} \label{Quantization_section}

We are now ready to state the main result of this section.

\begin{theorem} \label{thm:quantization}
Poisson BF theory on $\CC^{2n} \times \RR^m$ admits a translation invariant quantization in each of the following cases:
\begin{itemize}
\item[(a)] $m$ is even,
\item[(b)] $m = 1, 3$ or $5$.
\end{itemize}
\end{theorem}

\begin{remark}
We speculate that the quantization exists (and is unique) for all values of $m$, though we do not prove that here, see Remark \ref{general_quantization_remark}.
\end{remark}

Since Poisson BF theory is of BF type, a consequence of Theorem \ref{thm:oneloop} is that in order to prove this result we need to show that the one-loop anomaly $\Theta_{n,m}$ is cohomologically trivial in the cases mentioned.  
This anomaly is a degree one cocycle in the local cohomology of the Lie algebra $\cL_{n,m}$.  Since we are using the standard Poisson structure on $\CC^{2n}$, the classical theory, as well as the prequantization we have just constructed, is translation invariant.
Thus, we study quantizations that are also translation invariant.  Let us summarize the steps we will follow in order to prove Theorem \ref{thm:quantization}.
\begin{enumerate}
 \item We begin by proving the theorem for even $m \ge 2$.  For such theories we can use a very abstract argument to show that the anomaly vanishes, by realizing it as a deformation of a completely holomorphic theory, then using structural results from \cite{BWhol}.  This is Lemma \ref{even_m_lemma}.
 \item Next, we will begin to study, in detail, the local cohomology of $\cL_{n,m}$.  This local Lie algebra admits a map to the Lie algebra $\mr{Vect}^{\mr{hol}}_{\mr{symp}}(\CC^{2n})$ of  holomorphic symplectic vector fields.  In Lemma \ref{symp_vf_lemma} we will show that the anomaly class $[\Theta_{n,m}]$ can be lifted to a cocycle for $\mr{Vect}^{\mr{hol}}_{\mr{symp}}(\CC^{2n})$.
 \item Now, this is useful, because we can relate the cohomology of \emph{this} Lie algebra to something more mathematically familiar.  We show in Proposition \ref{ham_prop} that the Lie algebra cohomology of $\mr{Vect}^{\mr{hol}}_{\mr{symp}}(\CC^{2n})$ is equivalent -- up to a shift -- to the cohomology of an infinite-dimensional Lie algebra $\fh_{2n}$ of formal Hamiltonian vector fields, studied in work of Gelfand, Kalinin and Fuchs.  We have therefore reduced our calculation of $[\Theta_{n,m}]$ to the calculation of a cohomology class in $\fh_{2n}$.
 \item Even better, the cohomology of $\fh_{2n}$ is naturally graded, and the non-positively graded piece is relatively well understood.  In Lemma \ref{weight_lemma} we show that $[\Theta_{n,m}]$ lives in the weight 0 summand $\mr H^{4n+m+1}_{(0)}(\fh_{2n})$.
 \item Finally, we compute this weight 0 cohomology group, and show in Lemma \ref{ham_cohomology_lemma} that it vanishes when $m = 0,1,3$ or 5, thus proving the Theorem.
\end{enumerate}

So, we will first address the case when $m$ is even and greater than zero, in which case the anomaly vanishes for structural reasons.

\begin{lemma} \label{even_m_lemma}
If $m \ge 2$ is even then the class $[\Theta_{n,m}]$ of the 1-loop anomaly vanishes.
\end{lemma}

\begin{proof}
Suppose $m \ge 2$ is even, and let $r = \frac m2$. We can equip $\RR^m = \CC^{r}$ with its standard complex structure.
Notice that we can decompose the fields as follows:
\begin{align*}
A  = \sum_{k=0}^r A_k & \ \text{ in }\  \bigoplus_{k=0}^r \Omega^{0,\bu} (\CC^{2n}) \, \Hat{\otimes} \, \Omega^{k, \bu} (\CC^{r}) [1] \\
B = \sum_{k=0}^r B_k & \ \text{ in }\  \bigoplus_{k=0}^r \Omega^{2n,\bu} (\CC^{2n}) \, \Hat{\otimes} \, \Omega^{k, \bu} (\CC^{r}) [2n+m-2] .
\end{align*}
Using this decomposition, the action is of the form
\begin{equation}\label{holdef}
\int_{\CC^{2n} \times \CC^{r}} B \wedge \bigg( \left(\dbar + \dbar_{\CC^{r}}\right) A + \frac12 \{A , A\}_\Pi \bigg) \; + \; \sum_{k =1}^{{r}} \int_{\CC^{2n} \times \CC^{r}} B_k \wedge \partial_{\CC^{r}} \left(A_{r - k - 1} \right)
\end{equation}

Here, $\dbar$ continues to denote the Dolbeault operator on $\CC^{2n}$ and now $\partial_{\CC^{r}}, \dbar_{\CC^{r}}$ are the holomorphic and anti-holomorphic Dolbeault operators on $\CC^{r}$. 
Forgetting about the second term, the first term describes a purely holomorphic theory in the sense of \cite{BWhol}. 
(In fact, this describes Poisson BF theory for the holomorphic Poisson manifold $\CC^{2n+{r}}$ where we choose the trivial Poisson bivector on the last ${r}$ coordinates.)

Furthermore, the second term in (\ref{holdef}) involves only {\em holomorphic differential operators} with respect to our fixed complex structure. 
Thus we can view Poisson BF theory as a (translation invariant) holomorphic field theory on $\CC^{2n + r}$. As a consequence of \cite[Proposition 4.4]{BWhol} (see \cite[Lemma 7.2.7]{BCOV1} or \cite[Lemma B.1]{GwilliamWilliamsKM}, for a similar calculation) the anomaly $\Theta_{n,m}$ for such a theory is necessarily a sum of local functionals of the form
\begin{equation}\label{holanomaly}
\int_{\CC^{2n+r}} \left(D_{i_0} A\right) \partial \left(D_{i_1} A\right) \cdots \partial \left(D_{i_{2n}} A \right) \partial_{\CC^{r}} \left(D_{i_{2n+1}} A\right) \cdots \partial_{\CC^{r}} \left(D_{i_{2n+r}} A\right) .
\end{equation}
Here, $\partial$ denotes the holomorphic de Rham differential on $\CC^{2n}$ and the $D_{i_j}$'s are all translation invariant holomorphic differential operators. 

Each of these elements are considered as cochains in the cochain complex of local functionals equipped with the classical BV differential $\{S,-\}_{\mr{BV}}$.
In this case, the BV differential is of the form
\begin{equation}\label{holbv}
\{S, -\}_{\mr{BV}} = \dbar + \dbar_{\CC^{r}} + \d_{\Pi} + \partial_{\CC^{r}}
\end{equation}
Here, $\d_{\Pi}$ is the Chevalley--Eilenberg differential for the Poisson bracket $\{-,-\}_\Pi$.

We consider the spectral sequence converging to the local cohomology of $\cL_{n,m}$ induced by the filtration on the $A$-fields
\[
F^k = \Omega^{0,\bu} (\CC^{2n}) \, \Hat{\otimes} \, \Omega^{\geq k, \bu} (\CC^{r}) [1] .
\]
The $E_1$-page is given by the cohomology with respect to the first three terms in (\ref{holbv}). 

Assuming a class $\Theta_{n,m}$ which is a sum of functionals of the form (\ref{holanomaly}) survives to this page, we will show that it is rendered exact by the next term in the spectral sequence. 
By the formula, we see that at least one holomorphic derivative from each of the directions in $\CC^{2n + r}$ necessarily appear. 
For simplicity, consider the direction $z_{r}$.
By this observation, we can write
\[
\Theta_{n,m} = \frac{\partial}{\partial z_{{r}}} \Theta_{n,m}'
\]
for some local functional $\Theta_{n,m}'$ which is also of degree $+1$.
Finally, notice that the class $\iota_{\frac{\partial}{\partial z_{r}}} \Theta_{n,m}'$ satisfies
\[
\partial_{\CC^{r}} \left(\iota_{\frac{\partial}{\partial z_{r}}} \Theta_{n,m}'\right) = \Theta_{n,m}
\]
Here, if $X$ is a holomorphic vector field, the operator $\iota_{X}$ denotes the operator induced from the contraction $A \mapsto \iota_X A$. 
This shows that on the $E_\infty$-page all such functionals $\Theta_{n,m}$ become trivial and the result follows.
\end{proof}

\subsection{The Anomaly as a Gelfand--Fuchs Cocycle}

For the remaining cases of Theorem \ref{thm:quantization}, we must understand the cohomology of local functionals more explicitly. 

To start, we interpret the one-loop anomaly $\Theta_{n,m}$ as a local cocycle in a slightly different local Lie algebra. 

Consider the sheaf of Lie algebras $\cO^{\rm hol} (\CC^{2n})$ which is equipped with the Poisson bracket coming from the standard symplectic structure on $\CC^{2n}$. 
There is a short exact sequence of sheaves of Lie algebras
\begin{equation}\label{eqn:ses}
0 \to \ul{\CC} \to \cO^{\rm hol}(\CC^{2n}) \to {\rm Vect}^{\rm hol}_{\rm symp}(\CC^{2n}) \to 0
\end{equation}
where ${\rm Vect}^{\rm hol}_{\rm symp}(\CC^{2n})$ denotes the sheaf of holomorphic vector fields preserving the symplectic structure on $\CC^{2n}$. 
Here the right-most map associates to a function its Hamiltonian vector field. 

Of course, none of the sheaves in the above short exact sequence are local Lie algebras.
For $\cO^{\rm hol}(\CC^{2n})$ we know that $\Omega^{0,\bu}(\CC^{2n})$ provides a free resolution and hence a presentation as a local Lie algebra.
Similarly, for ${\rm Vect}^{\rm hol}_{\rm symp}(\CC^{2n})$ we have the following presentation as a local Lie algebra
\[
\begin{array}{ccccc}
& & \ul{0} & & \ul{1} \\
\cT_{\rm symp}(\CC^{2n}) & = & \Omega^{0,\bu}(\CC^{2n} , {\rm T}) & \xto{L_{(-)} \omega} & \Omega^{\geq 2, \bu} (\CC^{2n}) .
\end{array} 
\]
Here, ${\rm T}$ denotes the holomorphic tangent bundle and $\Omega^{0,\bu}(\CC^{2n}, {\rm T})$ is its Dolbeault resolution equipped with the $\dbar$ operator. 
Also $\Omega^{\geq 2, \bu}(\CC^{2n})$ is a Dolbeault model for the sheaf of closed two-forms, it is equipped with the differential $\dbar + \partial$. 
Finally, the indicated differential sends a vector field $X$ to the Dolbeault form $L_X \omega$, the Lie derivative of $\omega$ by $X$.  

The sheaf $\cT_{\rm symp}(\CC^{2n})$ becomes a local Lie algebra on $\CC^{2n}$ utilizing the standard Lie bracket of vector fields together with the natural action of vector fields on Dolbeault forms. 
Moreover, the process of taking a Hamiltonian vector field determines a map of local Lie algebras $\Omega^{0,\bu}(\CC^{2n}) \to \cT_{\rm symp}(\CC^{2n})$, providing a resolved version of the map in (\ref{eqn:ses}). 
Upon tensoring with de Rham forms on $\RR^m$ we obtain a map of local Lie algebras
\begin{equation}\label{eqn:resolved}
\cL_{n,m} \to \cT_{\rm symp}(\CC^{2n}) \; \Hat{\otimes} \; \Omega^{\bu}(\RR^m) .
\end{equation}

By Lemma \ref{lemma:Atype}, we a priori only know that the anomaly lives in $\cloc^\bu(\cL_{n,m})$. 
In fact, we have the following.

\begin{lemma} \label{symp_vf_lemma}
The anomaly $\Theta_{n,m}$ lifts to a local cocycle of degree $+1$ for $\cT_{\rm symp}(\CC^{2n}) \; \Hat{\otimes} \; \Omega^{\bu}(\RR^m)$ along the cochain map
\[
\cloc^\bu\left(\cT_{\rm symp}(\CC^{2n}) \; \Hat{\otimes} \; \Omega^{\bu}(\RR^m)\right) \to \cloc^\bu(\cL_{n,m})
\]
induced from (\ref{eqn:resolved}). 
\end{lemma}

\begin{proof}
The result actually follows from a similar statement at the level of the classical action. 
For this proof, we denote by $\cT_{n,m}$ the local Lie algebra $\cT_{\rm symp}(\CC^{2n}) \; \Hat{\otimes} \; \Omega^{\bu}(\RR^m)$.
First, we note that the natural contragradient action of an element of $\cL_{n,m}$ on $\cL^!_{n,m}$ is zero if the element is a constant function. 
Thus $\cL^!_{n,m} [-3]$ descends to a module for $\cT_{n,m}$. 
The map (\ref{eqn:resolved}) determines a cochain map
\[
\cloc^\bu \left(\cT_{n,m} \ltimes \cL_{n,m}^! [-3] \right) \to \cloc^\bu \left(\cL_{n,m} \ltimes \cL_{n,m}^![-3]\right) .
\]

Recall, the classical action decomposes as $S = S_{\rm free} + I$ where $I$ is a local cocycle in $\cloc^\bu (\cL_{n,m} \ltimes \cL_{n,m}^![-3])$. 
Since $I$ is identically zero if one of the $A$-inputs is a constant function, it lifts to a local cocycle in $\cloc^\bu \left(\cT_{n,m} \ltimes \cL_{n,m}^! [-3] \right)$. 

Proceeding just as in the proof of Lemma \ref{lemma:Atype}, we see that the one-loop anomaly only depends on $\cT_{n,m}$ and the result follows. 
\end{proof}

We will utilize a description from \cite{BWgf} of the local cohomology of 
\[
\cT_{\rm symp}(\CC^{2n}) \; \Hat{\otimes} \; \Omega^{\bu}(\RR^m)
\]
in terms of the ordinary Lie algebra cohomology of the fiber of this local Lie algebra at zero. 
For this, we first introduce the following Lie algebra of formal Hamiltonian vector fields on the formal (holomorphic) $2n$-disk.

\begin{definition}
The Lie algebra of formal Hamiltonian vector fields $\ham_{2n}$ has underlying $\CC$-vector space
\[
\ham_{2n} = \CC[[p_1, \ldots p_n, q_1, \ldots, q_n]] \; / \; \CC 
\]
On linear elements, the bracket is given by the formula $[p_i, q_j] = \delta_{ij}$. 
\end{definition}
Note that $\ham_{2n}$ is equivalent to the subalgebra of the Lie algebra of all formal vector fields on the $2n$-disk which preserve the standard symplectic structure.

We can now state the concrete relationship between local cohomology and the Lie algebra cohomology of Hamiltonian vector fields. 
Given any symplectic vector field $X$ on $\CC^{2n}$ we can take its Taylor expansion at zero to get a formal Hamiltonian vector field $j_0 (X) \in \fh_{2n}$. 
At the level of cochains this determines a map $j_0^* \colon \clie^\bu(\fh_{2n}) \to \clie^\bu(J (\cT^{\rm symp}(\CC^{2n})))$. 
Here $J(-)$ denotes the bundle of $\infty$-jets.
In \cite{BWgf} it is shown via certain ``descent equations" how to extend this map to the local cohomology to give a cochain map
\[
\clie^\bu(\fh_{2n}) [4n] \to \cloc^\bu(\cT^{\rm sym}(\CC^{2n})) .
\]
In fact, it is shown that this map is a quasi-isomorphism. 
The shift down by $4n$ is related to the {\em real} dimension of $\CC^{2n}$. 

Upon tensoring with the de Rham complex on $\RR^{m}$, which of course does not contribute to cohomology, we obtain the following analogous result.
The proof follows the outline of the proof of the main result of\cite{BWgf} closely, where the case of {\em all} holomorphic vector fields is considered. The modifications in the Hamiltonian case are specified in the proof below.

\begin{prop} \label{ham_prop}
There is a quasi-isomorphism 
\[
\cred^\bu(\fh_{2n}) [4n +m] \xto{\iso} \cloc^\bu \left( \cT_{\rm symp}(\CC^{2n}) \, \Hat{\otimes}\,\Omega^\bu(\RR^m) \right)
\]
Therefore the anomaly cocycle $\Theta$ can be understood as a cocycle in 
\[
\clie^{4n+m+1}(\fh_{2n}) .
\]
\end{prop}
\begin{proof}
Consider the Gelfand--Kazhdan pair $\left(\fh_{2n} \oplus \fw_m , \Sp(2n) \times \GL(m) \right)$, where $\fw_m$ denotes the Lie algebra of formal vector fields on $\RR^m$.  Given any module $M$ for this pair, techniques of formal geometry define a $D$-module $\cM$ on any product manifold of the form $X \times M$ where $X$ is a holomorphic symplectic manifold and $M$ is a smooth $m$-manifold. 
For the module $\cred^\bu(\fh_{2n})$ this $D$-module is equivalent to $\cred^\bu\left(J \cT_{\rm symp}(X)\right)$. 
Furthermore, since $\cred^\bu(\fh_{2n})$ is acted upon by the Gelfand--Kazhdan pair in a homotopically trivial way, this $D$-module is trivial and hence its de Rham complex is quasi-isomorphic to
\[
\Omega^{\bu} (X \times M) \otimes \cred^\bu(\fh_{2n}) .
\]
Finally, by \cite[Lemma 3.5.4.1]{Book2} we know that this de Rham complex is equivalent to the shift of the local cohomology $\cloc^\bu\left(\cT_{\rm symp}(X) \Hat{\otimes} \Omega^\bu(\RR^n) \right)[-4n-m]$. 
The result follows by applying this argument to the special case $X = \CC^{2n}$ and $M = \RR^m$.
\end{proof}

We have whittled down our understanding of the anomaly $\Theta_{n,m}$ to a description of the Lie algebra cohomology of Hamiltonian vector fields on the formal disk.
Unfortunately, a complete description of the cohomology of the Lie algebra $\ham_{2n}$ is unknown. Partial results have appeared in the works \cite{GKF, GuilleminShnider, Perchik, Kontsevich}. The first step in order to obtain such a partial description is to take advantage of a natural grading on $\fh_{2n}$ which we will refer to as the {\em weight}.  In fact, there is a bigrading on $\fh_{2n}$ in which the element $p_i^{k+1} q_j^{\ell+1}$ is homogenous of weight $(k, \ell)$.  The cohomology is concentrated in the diagonal piece of this bigrading, so it's enough to consider only the diagonal $\ZZ$-grading in which $p_i^{k+1} q_j^{\ell+1}$ has weight $k + \ell$.  
The bracket respects this grading and hence the Lie algebra cohomology admits a decomposition
\[
\mr H^\bu(\fh_{2n}) = \bigoplus_{j \in 2\ZZ} \mr H^\bu_{(j)} (\fh_{2n}) 
\]
where $H^\bu_{(j)} (\fh_{2n})$ is the cohomology of the weight $j$ subcomplex of the Lie algebra cohomology.

\begin{remark}
Our grading convention differs slightly from the one used in the reference \cite{GKF}, but agrees with the convention in \cite{Perchik}.
For instance, the weight $k$ cohomology, as defined in \cite{GKF}, agrees with $\mr H^\bu_{(2k)}(\fh_{2n})$ in our convention. 
\end{remark}

We return to the characterization of the anomaly cocycle.
The class $[\Theta_{n,m}]$ of the anomaly decomposes according to this weight decomposition of the Lie algebra cohomology group $\mr H^{4n+m+1}(\fh_{2n})$. 
In fact, we have the following.

\begin{lemma} \label{weight_lemma}
The anomaly class $[\Theta_{n,m}]$ lies in the weight zero summand $\mr H_{(0)}^{4n+m+1}(\fh_{2n})$ of the Lie algebra cohomology. 
\end{lemma}

\begin{proof}
We first assign a weight grading to the classical BV theory. 
Let $E$ be the holomorphic Euler vector field on $\CC^{2n}$
\[E = \sum_i \left(z^i \frac{\partial}{\partial z^i} + w^i \frac{\partial}{\partial w^i} \right),\]
where $z^i, w^i$ are holomorphic Darboux coordinates on $\CC^{2n}$.  We say that an element $A \in \mc L$ is of {\em weight} $j \in \ZZ$ if 
\[L_E (A) = (j + 2) A .\]
Notice that the differential and the bracket on the local Lie algebra $\mc L$ are of weight zero with respect to this grading. 

To define the grading on the cotangent theory to the local Lie algebra $\mc L$ we say an element $B \in \mc L^!$ is of weight $k \in \ZZ$ if
\[L_E (B) = (k-2) B .\]
With respect to this grading, the action $S$ of Equation (\ref{eqn:action2}) is of weight 0.  

Recall that the anomaly is characterized as the obstruction to solving the quantum master equation. 
Since the classical BV differential is automatically weight zero, we simply need to check that the BV Laplacian is weight zero. 
This is equivalent to checking that the shifted symplectic pairing
\begin{align*}
\cL_c \times \cL_c^! & \to \CC \\
A \otimes B & \mapsto \int_{\CC^{2n} \times \RR^m} A \wedge B
\end{align*}
is of weight zero. 
Observe that the pairing between a compactly supported $A$-field and a $B$-field can be written as
\[
\int_{\CC^{2n} \times \RR^m} A \wedge B = \int_{\CC^{2n} \times \RR^m} \<\Pi \otimes A, \omega \otimes B\>
\]
where on the right hand side we use the linear pairing between $\wedge^2 T_{\CC^{2n}}$ and $K_{\CC^{2n}}$. 
Since this linear pairing is manifestly weight zero, the claim follows. 

\end{proof}

\begin{remark}
We can provide a geometric description of the weight assignments used in proof of the previous lemma.
Consider the standard Poisson bivector $\Pi = \sum_i \partial_{z^i} \wedge\partial_{w^i}$.
It determines an isomorphism
\[
\Pi \otimes ( \cdot ) \colon \Omega^{0,\bu}(\CC^{2n}) \; \Hat{\otimes} \; \Omega^{\bu}(\RR^m) \xto{\cong} {\rm PV}^{2,\bu}(\CC^{2n}) \;\Hat{\otimes}\; \Omega^\bu(\RR^m) .
\]
The weight on $A$-fields agrees with pulling back the standard dilation on $\CC^{2n}$ acting on elements on the right hand side of this isomorphism. 
That is, $A$ has weight $j$ if and only if $L_E(\Pi \otimes A) = j (\Pi \otimes A)$. 
Similarly, if $\omega = \Pi^{-1}$, then we have an isomorphism
\[
\omega \otimes (\cdot ) \colon \Omega^{2,\bu}(\CC^{2n}) \; \Hat{\otimes} \; \Omega^{\bu}(\RR^m) \xto{\cong} \Omega^{0,\bu}(\CC^{2n}, K_{\CC^{2n}}^{\otimes 2}) \;\Hat{\otimes}\; \Omega^\bu(\RR^m) .
\]
Then, $B$ has weight $k$ if and only if $L_E(\omega \otimes B) = k (\omega \otimes B)$. 
\end{remark}

We now argue the vanishing of the one-loop anomaly of the theory on $\CC^{2n} \times \RR^m$ by collecting facts about the known cohomology of Hamiltonian vector fields.
In fact, a result of Gelfand, Kalinin, and Fuchs \cite{GKF} gives a complete description of the non-positive weight part of the cohomology of Hamiltonian vector fields $\ham_{2n}$.  
The description uses the Hochschild--Serre spectral sequence associated to the subalgebra $\sp(2n) \sub \ham_{2n}$.

Explicitly, the $E_2$-page of this spectral sequence is
\[
E_2^{i,j} = \mr H^i(\fh_{2n}, \sp(2n)) \otimes \mr H^j(\sp(2n))
\]
and the spectral sequence converges to $\mr H^{i+j}(\fh_{2n})$.  
In the remainder, we let $H^\bu_{(\leq 0)}$ denote the non-positiive weight part of the cohomology. 

\begin{theorem}[{\cite{GKF}[Theorem 2]}]
The non-positive weight part of the relative Gelfand-Fuchs cohomology $\mr H^i_{(\le 0)}(\mf{h}_{2n}, \sp(2n)) \otimes \mr H^j(\sp(2n))$ is isomorphic to the algebra
\[\CC[\Gamma, \Psi_1, \ldots, \Psi_n]/I\]
where $\Gamma$ has degree $2n-1$ and weight $-1$, $\Psi_i$ has degree $4i$ and weight $0$, and the ideal $I$ is generated by the elements $\Gamma^k\Psi_1^{k_1} \cdots \Psi_n^{k_n}$ where $k + k_1 + 2k_2 + \cdots nk_n > n$.  

Furthermore, in the Hochschild--Serre spectral sequence, the standard generators $h_i$ of $\mr H^\bullet(\sp(2n))$ of degrees $3, 7, \ldots, 4n-1$ map under transgression to the generators $\Psi_1, \ldots, \Psi_n$.
\end{theorem}

This implies the following result.
\begin{lemma} \label{ham_cohomology_lemma}
The cohomology group $\mr H^{4n+k}_{(0)}(\mf{h}_{2n}) = 0$ if $k = 1,2, 4$ or $6$, for all $n$.  
\end{lemma}

\begin{proof}
Classes of weight $0$ are generated by monomials of the form $a = \Psi_1^{k_1} \cdots \Psi_n^{k_n} h_1^{\ell_1} \cdots h_n^{\ell_n}$, where each $\ell_i = 0$ or 1.  If $c$ does not include $h_i$ then $a$ is in the image of the transgression map, and therefore vanishes, so from now on we will assume that at least one $\ell_j \ne 0$.  If such a class survives to the $E_\infty$ page of the spectral sequence it must in particular be in the kernel of the transgression map.  

Now, the classes $\Psi_i$ have degree 0 mod 4 and $h_i$ have degree 3 mod 4.  Let us consider the possible degrees of classes involving different numbers of $h_i$.  Applying the transgression to $c$ gives a sum of elements of the form $\Psi_1^{k_1} \cdots \Psi_j^{k_j+1} \cdots \Psi_n^{k_n} h_1^{\ell_1} \cdots \hat h_j \cdots h_n^{\ell_n}$, where for each summand $\sum_{i=1}^n i k_i + j > n$.  The element $\Psi_1^{k_1} \cdots \Psi_n^{k_n} h_j$ therefore has degree at least $4n$
\begin{enumerate}
 \item If only one $\ell_j \ne 0$ then $c$ has degree 3 mod 4.
 \item If two $\ell_j \ne 0$ then $c$ has degree 2 mod 4.  Because $\Psi_1^{k_1} \cdots \Psi_n^{k_n} h_j$ has degree at least $4n$, a class of the form $\Psi_1^{k_1} \cdots \Psi_n^{k_n} h_{j_1} h_{j_2}$ has degree at least $4n + 3$.
 \item If three or more $\ell_j \ne 0$ then, likewise $c$ has degree at least $4n + 3 + 7 = 4n + 10$.
\end{enumerate}
In particular, there are no non-zero classes in degree $4n+1, 4n+2, 4n+4$ or $4n+6$.
\end{proof}

Now, we can use this cohomology calculation to immediately prove the main theorem of this section.
\begin{proof}[Proof of Theorem \ref{thm:quantization}]
By Lemma \ref{weight_lemma}, the class $[\Theta_{n,m}]$ of the one-loop anomaly lies in the weight 0 summand $\mr H_{(0)}^{4n+m+1}(\mf{h}_{2n})$.  
By Lemma \ref{ham_cohomology_lemma} this cohomology group vanishes, and therefore so does the anomaly.
\end{proof}

\begin{remark}
The same calculation shows that this quantization is unique among weight 0 quantizations in the cases where $m=0,1,2,4$ or $6$.  Indeed, deformations of a quantum field theory are controlled by classes in $\mr H^0_{\mr{loc}}(\mc L_{n,m})$.  We have demonstrated that the weight 0 part of this cohomology vanishes if $m=0,1,2,4$ or $6$.
\end{remark}

\begin{remark} \label{general_quantization_remark}
At this point it seems natural to speculate that all Poisson BF theories are anomaly free.  The first example to which our methods do not apply is the 17-dimensional Poisson BF theory on $\CC^4 \times \RR^9$.  The anomaly here is given by a weight 0 class in $\mr H_{(0)}^{18}(\mf h_4) \iso \CC^2$ (this is the smallest non vanishing weight 0 even cohomology group of a Lie algebra $\mf h_{2n}$, it is generated by the classes $\Psi_1^2h_1h_2$ and $\Psi_2h_1h_2$.).
\end{remark}

\section{Occurence as Twisted Supergravity} \label{sugra_section}

The BV theories we have introduced so far have conjectural connections to string theory and supergravity, using the theory of twisted supergravity as discussed by Costello and Li \cite{CostelloLiSUGRA}.  In this section we provide a brief survey of these relationships, but these ideas will not be used in the remainder of the paper.

\begin{conjecture}
The twist of 5d $\mc N = 1$ supergravity on $\RR^5$ is equivalent to Poisson BF theory on $\CC^2 \times \RR$ where $\CC^2$ carries its canonical holomorphic symplectic structure. 
\end{conjecture}

\begin{remark}
This conjecture can be deduced from a description of the twist of $M$-theory on the $11$-manifold $\RR \times \CC^3 \times \CC^2$ discussed in \cite{PhilSurya} upon dimensional reduction along the Calabi--Yau three-fold $\CC^3$. 
More generally, one can study $M$-theory on a geometry of the form $\RR \times X \times \CC^2$, where $X$ is an arbitrary Calabi--Yau three-fold.
The reduction along $X$ is expected to produce $\cN=1$ supergravity coupled to a $\cN=1$ vector multiplet with gauge group $\U(1)$ and a $\cN = 1$ hypermultiplet valued in the symplectic vector space ${\rm H}^3(X)$. 
We thank Kevin Costello for clarifying this point.
\end{remark}

The proof of this conjecture starting from an explicit description of 5d $\mc N=1$ supergravity is currently joint work with Ingmar Saberi. Let us explain some evidence for this conjecture.  First, we recall the form of the twist of 5d $\cN=1$ super Yang--Mills theory on $\RR^5$, which has been worked out in \cite{ESW}. Recall that 5d $\cN=1$ gauge theory is defined for any Lie algebra $\fg$ together with a representation $V$.
The twist of this theory on $\CC^2 \times \RR$ is equivalent to the cotangent theory to the local dg Lie algebra
\[
\Omega^{0,\bu}(\CC^2) \; \Hat{\otimes} \; \Omega^\bu(\RR) \otimes \left(\fg \ltimes V[-1] \right) .
\]
which we will denote by $\cL_{\rm SYM}(\fg, V)$. 

Given this description, it is clear how $\cL_{\rm SYM}(\fg, V)$ is a module for the local dg Lie algebra 
\[
\cL_{2,1} = \Omega^{0,\bu}(\CC^2) \; \Hat{\otimes} \; \Omega^\bu(\RR)
\]
underlying Poisson BF theory. 
Indeed, if $A \in \cL_{2,1}$ and $\gamma \in \cL_{\rm SYM}(\fg, V)$ the action is through the Poisson bracket $\{A, \gamma\}_\Pi$.

In this way, there is a natural way to ``couple" Poisson BF theory to the twist of super Yang--Mills theory.
One considers the cotangent theory to the local dg Lie algebra
\[
\cL_{2,1} \ltimes \cL_{\rm SYM}(\fg, V) .
\]

The dimensional reduction of 5d $\cN=1$ SYM along $\RR$ is 4d $\cN=2$ SYM. 
For pure gauge theory, this means that the dimensional reduction of 5d $\cN=1$ SYM is equivalent to 4d $\cN=1$ SYM coupled to an adjoint-valued 4d $\cN=1$ chiral multiplet. 

Following this line of reasoning, the dimensional reduction of 5d $\cN=1$ supergravity is 4d $\cN=2$ supergravity, which we can hope to further decompose in terms of 4d $\cN=1$ supersymmetry. 
The reduction of our 5d theory along $\RR$ is equivalent to the cotangent theory to the local dg Lie algebra 
\[
\Omega^{0,\bu}(\CC^2)[\ep] = \Omega^{0,\bu}(\CC^2) \ltimes \ep \Omega^{0,\bu}(\CC^2) 
\]
where $\ep$ is a parameter of cohomological degree $+1$. 
The Lie structure is similar to that of $\cL_{2,0}$: we take the dg Lie algebra $\cL_{2,0} = \Omega^{0,\bu}(\CC^2)$ with its Poisson bracket and tensor it with the graded ring $\CC[\ep]$. 

We can decompose the fields of this 4d theory as $A + \ep \gamma$, where $A, \gamma \in \Omega^{0,\bu}(\CC^2)$. 
Denote the anti-fields by $B, \beta$.
The action functional reads
\[
\int_{\CC^2} B\wedge \left(\dbar A + \frac{1}{2} \{A,A\}_\Pi \right) + \int_{\CC^2} \beta \wedge \dbar \gamma + \int_{\CC^2} \beta \wedge \{A, \gamma\}_\Pi .
\]
The first term we recognize as holomorphic Poisson BF theory on $\CC^2$, described by the local dg Lie algebra $\cL$ of Definition \ref{dfn:localLie1}.
The second term is that of the free $\beta\gamma$ system on $\CC^2$, this is equivalent to the twist of a single 4d $\cN=1$ chiral multiplet, see \cite{ESW}. 

This leads us to the following 4d analogue of our conjecture on twisted supergravity.

\begin{conjecture}
The twist of pure 4d $\mc N = 1$ supergravity on $\RR^4$ is equivalent to Poisson BF theory on $\CC^2$ with its standard holomorphic symplectic structure; this is the cotangent theory to the local Lie algebra $\cL_{2,0} = \Omega^{0,\bu}(\CC^2)$. 
\end{conjecture}
 
\begin{remark}
Let $Y$ be a G2 manifold.
This conjecture can also be deduced from the dimensional reduction of the conjectural twist of $M$-theory on $Y \times \CC^2$, given in \cite{PhilSurya}, along the G2 manifold $Y$. 
\end{remark}

\begin{remark}
In the previous section we have shown that Poisson BF theory on $\CC^2$ is non-anomalous. 
This is compatible with the conjecture above and the result of Alvarez-Gaum\'{e} and Witten in \cite{GaumeWittenSUGRA} that there are no pure gravitational anomalies in any theory of supergravity in dimension $4k$, where $k \in \ZZ$.
\end{remark}

\section{Factorization Algebras and Symplectic Symmetry} \label{factn_section}

\subsection{Factorization Algebras}

In the introduction we have mentioned the relationship of quantum field theory to factorization algebras. 
Thanks to the foundational work of \cite{Book2}, we know that to every quantum field theory one can associate a factorization algebra of ``observables".  

At the classical level, this idea is based on the following sequence of observations.
First, note that the definition of the space of fields $\cE$ of a theory is local on spacetime; it is given as the sheaf of smooth sections of a (graded) vector bundle. 
For such sheaves, the continuous dual $\cE^\vee$ carries the structure of a cosheaf. 
The natural product on the completed symmetric algebra $\cO (\cE) = \Hat{\Sym}(\cE^\vee)$ endows the factorization structure maps as in Definition \ref{dfn:fact}. 

Together with the BV bracket, the classical action functional $S$ determines a differential $\{S, -\}_{\mr{BV}}$ on $\cO(\cE)$.
\begin{definition}
The factorization algebra of \emph{classical observables} in the theory $(\mc E, S)$ (a factorization algebra valued in cochain complexes)  is defined to be
\[\Obs^{\rm cl} = \bigg(\cO(\cE) \; , \; \{S,-\}_{\mr{BV}} \bigg) .\]
\end{definition}

Consider the example of Poisson BF theory on $X \times M$, where $X$ is a complex manifold equipped with a holomorphic Poisson structure $\Pi$, so $\mc E = \mc L_{\Pi, M} [1] \oplus \mc L^!_{\Pi, M}[-2]$.  
For simplicity, denote $\cL = \cL_{\Pi,M}$ in what follows.

We first note that the continuous dual of $\mc E = \mc L [1] \oplus \mc L^![-2]$ is 
\[
\cE^\vee = \Bar{\cL}^!_c [-1]  \oplus \Bar{\mc L}_c [2] 
\]
where $\Bar{(-)}_c$ denotes the space of compactly supported distributional sections.
Hence, as a graded vector space the observables have the form $\Sym\left(\Bar{\cL}^!_c [-1]\right)  \otimes \Sym\left(\Bar{\mc L}_c [2]\right)$.  

The classical BV differential $\{S,-\}_{\mr{BV}}$ can be identified with the Chevalley--Eilenberg differential of the dg Lie algebra $\cL$ with values in the module $\Sym \left(\Bar{\cL}_c[2] \right)$. 
Thus, the classical observables have the form
\[
\Obs^{\rm cl} = \clie^\bu \bigg(\cL \; ; \; \Sym \left(\Bar{\cL}_c [2]\right) \bigg) .
\]

Fix an open set inside $X \times M$, which for simplcity we'll take of the form $U \times V$ where $U$ is an open subset of $X$ and $V$ is an open subset of $M$.   
Then we have the dg Lie algebra
\[
\cL (U \times V) = \Omega^{0,\bu}(U) \Hat{\otimes} \Omega^\bu(V) 
\]
with differential $\dbar + \d_{\rm dR}$ and bracket $\{-,-\}_\Pi$.
The observables supported on $U \times V$ take the form
\[ 
\Obs^{\rm cl} (U \times V) = \clie^\bu \bigg(\Omega^{0,\bu}(U) \Hat{\otimes} \Omega^\bu(V) \; ; \; \Sym \left(\Bar{\Omega}_c^{0,\bu}(U) \Hat{\otimes} \Bar{\Omega}_c^\bu(V) [2]\right) \bigg) .
\]

We now immediately observe the following.

\begin{prop}
The factorization algebra $\Obs^{\rm cl}$ is not Betti topological.
\end{prop}

\begin{proof}
Write $D_r(0) \sub \CC^{2n}$ for the polydisk around 0 of radius $r$.  Choose $r < R$, and any open subset $V$ of $\RR^m$, and consider the inclusion $D_r(0) \times V \inj D_R(0) \times V$.  The associated map on classical local observables is not an equivalence: this is dual to the observation that the map $\Omega^{0,\bu}(D_r(0)) \to \Omega^{0,\bu}(D_R(0))$ is not surjective on cohomology.
\end{proof}

\subsection{Translation and Dilation Actions} \label{translation_section}

Let us begin by recalling what it means for a group to act on a classical or quantum field theory, using the language of factorization algebras.  We will not include all details here, for a more thorough account see \cite[Chapter 4.8]{Book2} and \cite[Section 2]{ElliottSafronov}.  Let $G$ be a Lie group with Lie algebra $\gg$, and fix a smooth action $\rho$ of $G$ on the spacetime manifold $X$.  If $\obs$ is a factorization algebra on $X$, we can define a smooth action of $G$ on $\obs$ in the following way.

\begin{definition}[{\cite[Definition 2.11]{ElliottSafronov}}]
A \emph{smooth action} of $G$ on $\obs$ consists of an isomorphism $\alpha_g \colon \obs(U) \to \obs(\rho(g)(U))$ for every $g \in G$ and open $U \sub X$, satisfying the following conditions.
\begin{enumerate}
 \item $\alpha_{g_1g_2} = \alpha_{g_1} \circ \alpha_{g_2}$ for all $g_1, g_2$ in $G$.
 \item The map $\alpha_g$ commutes with the factorization structure, for all $g$ in $G$.
 \item For all collections of pairwise disjoint open subsets $U_1, \ldots, U_k$ of an open set $V$, the map
 \[m \colon \{(g_1, \ldots, g_k) \in G^k \colon g_k(U_k) \text{ are disjoint subsets of } V\} \to \hom\left(\bigotimes_{i=1}^k \obs(U_i), \obs(V)\right),\]
 defined by first acting by $(g_1, \ldots, g_k)$ then using the factorization structure, is smooth.
 \item There is an infinitesimal action $\rho \colon \gg \to \mr{Der}(\obs)$ of $\gg$ such that for all $w \in \gg$ and $i = 1, \ldots, k$,
 \[\dd_{w,i}m_{g_1, \ldots, g_k}(\mc{O}_1, \ldots, \mc{O}_k) \iso m_{g_1, \ldots, g_k}(\mc{O}_1, \ldots, \ldots, \rho(w)\mc{O}_i, \ldots, \mc{O}_k)\]
 where the map $\dd_{w,i}$ id the directional derivative on $G^k$ with respect to the tangent vector
 \[(0, \ldots, L_{g_i}(w), \ldots, 0) \in T_{g_1, \ldots, g_k}G^k,\]
 where the non-zero element is placed in the $i^{\mr{th}}$ slot.
\end{enumerate}
\end{definition}

The smooth action extends the infinitesimal action $\rho$ of the Lie algebra $\gg$ to a global action of the Lie group $G$.  With this definition in hand, we can now make precise what it means for a field theory to be de Rham topological.  We say a smooth action is \emph{de Rham} if this infinitesimal action is homotopically trivialized, in the following sense.

\begin{definition}[{\cite[Definition 2.18]{ElliottSafronov}}] \label{dR_action_def}
Define $\gg_{\mr{dR}}$ to be the dg Lie algebra with underlying cochain complex $\gg[1] \overset {\mr{id}} \to \gg$, where the degree zero Lie algebra $\fg$ acts on $\fg[1]$ by the adjoint representation. 

An action $\rho$ of a Lie algebra $\gg$ on a factorization algebra $\obs$ is \emph{de Rham} if $\rho$ is equipped with an extension to
\[\rho_{\mr{dR}} \colon \gg_{\mr{dR}} \to \mr{Der}(\obs).\]
A \emph{de Rham} action of a Lie group $G$ on $\obs$ is a smooth action of $G$ where the infinitesimal action of $\gg$ is extended to a de Rham action.
\end{definition}

\begin{definition}
We say a factorization algebra $\obs$ on $\RR^n$ is \emph{de Rham topological} if there is a de Rham action of $\RR^n$, where $\RR^n$ acts on itself by translations.
\end{definition}

The last general bit of background we need concerns a special type of action of a Lie algebra $\fg$ on the classical factorization algebra of observables $\Obs^{\rm cl}$. 
Recall that the BV bracket $\{-,-\}_{\mr{BV}}$ endows the shift of the cochain complex of local functionals $\oloc(\cE)[-1]$ with the structure of a dg Lie algebra.
The differential is given by $\{S,-\}_{\mr{BV}}$. 

\begin{dfn}
An \emph{inner action} of a dg Lie algebra $\fg$ on a classical field theory is a map of $L_\infty$ algebras
\[
\cP \colon \fg \to \oloc(\cE)[-1] .
\]
\end{dfn}

Giving an inner action on a theory is equivalent to prescribing an element
\[
\cP \in {\rm C}_{\rm red}^\bu(\fg) \otimes \oloc(\cE)[-1]
\]
satisfying the {\em equivariant classical master equation}
\begin{equation}\label{eqn:equivcme}
\d_{\fg} \cP + \{S, \cP\}_{\mr{BV}} + \frac{1}{2} \{\cP, \cP\}_{\mr{BV}} = 0
\end{equation}
where $\d_{\fg}$ denotes the Chevalley--Eilenberg differential for $\fg$. 

Notice that through the BV bracket $\{-,-\}_{\mr{BV}}$ the dg Lie algebra $\oloc(\cE)[-1]$ acts on the classical observables.  That is, there is a dg Lie map
\begin{align*}
\oloc(\cE)[-1] &\to {\rm Der}(\Obs^{\rm cl}) \\
\mc O &\mapsto \{\mc O, -\}_{\mr{BV}}.
\end{align*}
Therefore, any inner action by a Lie algebra $\fg$ determines an action of $\fg$ on the factorization algebra $\Obs^{\rm cl}$ as defined above.

Let us now focus attention on the example of Poisson BF theory.  
As in \S \ref{sec:quantization} we consider Poisson BF theory on $\CC^{2n}\times \RR^m$ where $\CC^{2n}$ is equipped with its standard symplectic structure.
We will begin by constructing a de Rham action of the Lie algebra of holomorphic functions which is equipped with the Poisson bracket. 
As above, $\cE_{n,m}$ denotes the fields of Poisson BF theory and $\cL_n$ denotes the dg Lie algebra $\Omega^{0,\bu}(\CC^{2n})$ resolving the Lie algebra of holomorphic functions on $\CC^{2n}$. 

\begin{theorem} \label{de_rham_holo_action_thm}
There is an inner action of $(\cL_{n})_{\rm dR}$ on Poisson BF theory on $\CC^{2n} \times \RR^m$. 
This induces a de Rham action of the Lie algebra $\mc L_n$ on the classical observables.
\end{theorem}

\begin{proof}
We will construct a (strict) map of dg Lie algebras 
\[
\cP_{\dR} \colon (\cL_n)_{\rm dR} = \cL_n \oplus \cL_n [1] \to \oloc(\cE_{n,m}) [-1]
\]
which we split up as a pair of linear maps $\cP_{\dR} = (\cP, \cQ)$. 

Given $\alpha \in \cL_{n}$ define the local functionals $\cP_\alpha, \cQ_\alpha \in \oloc(\cE_{n,m})$ by
\begin{align}
\cP_\alpha (A,B) & = \int B \wedge \{p^* \alpha, A\}_\Pi \\
\cQ_\alpha (A,B) & = \int B \wedge p^* \alpha . \label{eqn:Q}
\end{align}
Here, $p \colon \CC^{2n} \times \RR^m \to \CC^{2n}$ is the projection.
Together, these define the pair of linear maps $\cP \colon \alpha \mapsto \cP_\alpha$ and $\cQ \colon \alpha \mapsto \cQ_\alpha$. 
Note that if $\alpha$ is a $(0,k)$-form then $\cP_\alpha$ is a local functional of degree $k-1$ and $Q_\alpha$ is a local functional of degree $k-2$. 

The equivariant classical master equation (\ref{eqn:equivcme}) for $\cP_{\rm dR}$ is equivalent to the three equations
\begin{align*}
\d_{\cL_n} \cP + \{S, \cP\}_{\mr{BV}} + \frac{1}{2} \{\cP, \cP\}_{\mr{BV}} & = 0 \\
\dbar \cQ + \cP + \{S, \cQ\}_{\mr{BV}} & = 0  \\
\d_{\Pi} \cQ + \{\cP, \cQ\}_{\mr{BV}} & = 0 .
\end{align*}
The first equation is implied by the ordinary classical master equation $\{S, S\}_{\mr{BV}} = 0$.
(This simply says that the classical theory is equivariant for the dg Lie algebra $\cL_n$, which is clear.)

Next, let's consider the third equation.
Here, $\d_{\Pi}$ denotes the component of the Chevalley--Eilenberg differential for $(\cL_n)_{\rm dR}$ arising from the Poisson bracket. 
We apply the left-hand side to a pair of forms $\alpha \in \cL_n$ and $\beta \in \cL_n [1]$. 
Then $(\d_\Pi \cQ)(\alpha,\beta) = \cQ_{\{\alpha, \beta\}_\Pi}$
and 
\begin{align*}
\{\cP, \cQ\}_{\mr{BV}}(\alpha, \beta) = - \{\cP_\alpha, \cQ_\beta\}_{\mr{BV}} & = - \left\{\int B \wedge \{p^* \alpha, A\}_\Pi , \int B \wedge p^* \beta \right\} \\ 
& = - \int B \wedge\{p^* \alpha, p^* \beta\}_\Pi \\ & = - \cQ_{\{\alpha,\beta\}_\Pi} .
\end{align*}

Finally, we turn our attention to the second equation.
Decomposing the action into free and interacting summands, $S = S_{\rm free} + I$ where $I = \int B \wedge \{A, A\}_\Pi$, and recalling that $\{S_{\rm free} , -\}_{\mr{BV}} = \dbar$, we see that the second equation is equivalent to $\cP + \{I, \cQ\}_{\mr{BV}} = 0$. 
The operator $\{I, -\}_{\mr{BV}}$ is the component of the Chevalley--Eilenberg differential for $\cL_{n,m}$ which encodes the Lie bracket $\{-,-\}_\Pi$.
Evaluating this on an element $\alpha \in \cL_n$ we thus obtain
\begin{equation} \label{eqn:Q2}
\{I, \cQ\}_{\mr{BV}} (\alpha) = - \int B \wedge \{p^* \alpha, A\}_\Pi = - \cP_\alpha
\end{equation}
as desired. 
\end{proof}

Given any translation invariant vector field $v$ on $\CC^{2n} \times \RR^m$ we define the local functional
\begin{equation}\label{eqn:T}
\cT_v (A, B) = \int B \wedge L_{v} (A) 
\end{equation}
where $L_v (-)$ denotes the Lie derivative. 
This defines an inner action of infinitesimal translations 
\[\cT \colon \RR^{4n + m} \to \oloc(\cE_{n,m})[-1]\] 
sending $v \mapsto \cT_v$. 

Let us fix Darboux coordinates $\{z_i\}$ and $\{w_j\}$ on $\CC^{2n}$. 
The Lie algebra of holomorphic translations spanned by holomorphic derivatives in $z_i$ and $w_j$ admits a linear map to holomorphic functions via 
\[
\begin{array}{ccccccc}
\RR^{2n} & \to & \cO^{\rm hol}(\CC^{2n})\\
\frac{\partial}{\partial z_i} & \mapsto & w_i \\
\frac{\partial}{\partial w_j} & \mapsto & - z_j .
\end{array}
\]
Of course, this map is {\em not} a Lie map since $\{z_i, w_j \}_\Pi = \delta_{ij}$. 
However, since the Poisson bracket of the constant function $1$ with any holomorphic function is zero, we see that the composition 
\[
\RR^{2n} \to \cO^{\rm hol} (\CC^{2n}) \xto{\simeq} \Omega^{0,\bu}(\CC^{2n}) \xto{\cP} \oloc(\cE_{n,m}) [-1] ,
\]
where $\cP$ is defined in the proof of the previous theorem, {\em is} a map of dg Lie algebras. 
In fact, if $v = \frac{\partial}{\partial z_i}$ or $\frac{\partial}{\partial w_j}$ is a holomorphic translation invariant vector field, then $\cT_v = \cP_{w_i}$ or $-\cP_{z_j}$, respectively. 

\begin{corollary}
The classical Poisson BF theory on $\CC^{2n} \times \RR^m$ is de Rham translation invariant. 
\end{corollary}

\begin{proof}
First note that Poisson BF theory is translation invariant.  
That is, there is a smooth action of the group $\RR^{4n+m}$ by translations, which acts on the observables infinitesimally by the directional derivative.
The infinitesimal action is inner and defined by the Lie map $\cT \colon \RR^{4n+m} \to \oloc(\cE_{n,m})[-1]$ in (\ref{eqn:T}). 

We must describe an extension of this infinitesimal action to an action of $\RR^{4n+m}_{\mr{dR}}$.
We will do this by finding a potential $\cS$ such that the pair of maps
\[
(\cT, \cS) \colon \RR^{4n+m}_{\rm dR} \to \oloc(\cE_{n,m})[-1]
\]
defines an inner action by $\RR^{4n+m}_{\rm dR}$. 

Choose Darboux coordinates $\{z_i,w_j\}$ for $\CC^{2n}$ as above and denote by $\{t_k\}$ the coordinate on $\RR^m$. 

The potential for translations in the $\Bar{z}_i,  \Bar{w}_j$ and $t_k$ directions can be written as follows.  
Given a translation invariant vector field $v$ in the span of these directions, let $\iota_v$ denote the interior product with $v$, and define the local functional 
\[
\mc S_v (A,B) = \int B \wedge \iota_v(A).\]

The potential for translations in the holomorphic $z_i$ and $w_j$ directions has already been written down in the proof of Theorem \ref{de_rham_holo_action_thm}. 
We set
\[
\cS_{\frac{\partial}{\partial z_i}} = \cQ_{w_i} \quad , \quad \cS_{\frac{\partial}{\partial z_i}} = -\cQ_{z_j} 
\]
where $\cQ$ is defined in (\ref{eqn:Q}). 

We show that the equivariant classical master equation holds for the pair of maps $(\cT, \cS)$:
\begin{equation}
\d_{\RR^{4n+m}_{\rm dR}} (\cT + \cS) + \{S, \cT + \cS\}_{\mr{BV}} + \frac{1}{2} \{\cT + \cS, \cT + \cS\}_{\mr{BV}} = 0
\label{dR_CME}
\end{equation}
where $\d_{\RR^{4n+m}_{\rm dR}}$ is the Chevalley--Eilenberg differential for $\RR^{4n+m}_{\rm dR}$. 
Notice, first, that by translation invariance we have $\{S, \cT\}_{\mr{BV}} = 0$. 
Next, the potentials $\mc Q_v$ for $v$ mutually commute, and they commute with the infinitesimal translation action so the final term also vanishes. 

It suffices to show this equation holds upon applying any fixed translation invariant vector field to the left-hand side.
Applied to a fixed vector field $v$ we have
\[
\d_{\RR^{4n+m}_{\rm dR}} (\cT + \cS) (v) = \cT_v .
\]

For the second term in equation \ref{dR_CME}, suppose first that $v$ is in the span of translations in the $\Bar{z}_i,  \Bar{w}_j$ and $t_k$ directions. 
Then
\[
\{S, \cS\}_{\mr{BV}}(v) = - \left(\dbar + \d_{t} \right) \cS_v - \{I, \cS_v\}_{\mr{BV}} .
\] 
Here, $\d_t$ denotes the de Rham differential in the $\RR^m$ direction. 
Since $\iota_v$ is a derivation for the Poisson bracket $\{-,-\}_\Pi$ the term $\{I, \cS_v\}_{\mr{BV}}$ vanishes. 
Finally,
\[
(\dbar + \d_t) \cS_v = \int B \wedge [\dbar + \d_t, \iota_v] A = \cT_v
\]
by Cartan's formula, and the classical master equation \ref{dR_CME} follows.

Finally, suppose that $v$ is a holomorphic translation, say $\frac{\partial}{\partial z_i}$.
Then $\cT_v = \cP_{w_i}$ and by (\ref{eqn:Q2}) we have 
\[
\{S, \cS\}_{\mr{BV}}(v) = - \{S, \cQ_{w_i}\}_{\mr{BV}} = - \cP_{w_i},
\]
and again equation \ref{dR_CME} follows.
\end{proof}

\begin{remark}
More generally, on a holomorphic symplectic manifold the infinitesimal action of Theorem \ref{de_rham_holo_action_thm} will extend to a smooth de Rham action of the group of holomorphic symplectomorphisms.  On $\CC^{2n}$ such symplectomorphisms include holomorphic translations, and this action combines with the de Rham action of anti-holomorphic translations to make the theory de Rham topological.
\end{remark}

\begin{remark}
We can additionally describe a smooth action of the group $\RR_{>0}$ of dilations, that is, where $c \in \RR_{>0}$ acts on $\CC^{2n} \times \RR^m$ by simultaneously rescaling all the coordinate directions by $c$.  Infinitesimally, this action is described by the action of the Euler vector field on $\mc L_{n,m}$.  This action is not inner and not de Rham, since any de Rham translation invariant factorization algebra which is additionally de Rham dilation invariant is automatically Betti topological \cite[Proposition 3.38]{ElliottSafronov}, which is not the case for Poisson BF theory.
\end{remark}

\subsection{Factorization Algebra of Quantum Observables}

Due to issues of renormalization in the definition of a quantum field theory, as surveyed in \S \ref{sec:renorm}, the definition of the factorization algebra of {\em quantum observables} is much more subtle than the definition of classical observables.  To begin, one first defines a cochain complex of global observables. 
We provide a brief synopsis of the construction of the factorization algebra of quantum observables, but refer the reader to \cite{Book2} for complete details. 

We begin by fixing the data of a quantum field theory described by an effective family $\{I[L]\}$. 
A {\em global observable} $\mc O$ is an assignment of an $\hbar$-dependent functional on the space of fields
\[
O[L] \in \cO(\cE(M)) [[\hbar]]
\]
for each ``length scale" $L > 0$. 
The functionals $O[L], O[L']$ at different length scales $L < L'$ must be related by the renormalization group flow which $O[L'] = W_{L < L'} (O[L])$ (this condition also appears in the definition of an effective family as detailed in \S \ref{sec:renorm}).

The space of global observables is a cochain complex with differential 
\[
\d_L = Q + \{I[L], -\}_L + \hbar \triangle_L
\]
where $I[L]$ is the scale $L$ effective action and $\triangle_L$ is the scale $L$ BV Laplacian. 
The fact that renormalization group flow $W_{L<L'}$ intertwines the differentials $\d_L$ and $\d_{L'}$ turns this into a well-defined definition of the cochain complex of global observables, which we denote by $\Obs^{\rm q}(M)$. 

The next step is to define what a {\em local} observable is.
This is the most technical part of the definition. 
Given an open set $U \subset M$, one says that $\{O[L]\}$ is an element $\Obs^{\rm q} (U)$ if $O[L]$ is supported on $U$ in the limit $L \to 0$. 
Roughly, this means that for $L$ close to zero $O[L]$ has support approximately in the open set. 

The factorization product is described in a similar way to the classical case. 
It utilizes the commutative product on $\cO(\cE)[[\hbar]]$ together with renormalization group flow. 
To get a sense for the definition, let's consider the case of two disjoint open sets $U,V$ in $M$ and the factorization product
\[
m_{U,V} \colon \Obs^{\rm q}(U) \times \Obs^{\rm q} (V) \to \Obs^{\rm q}(U \sqcup V) .
\]
If $O = \{O[L]\} \in \Obs^{\rm q}(U)$ and $O' = \{O'[L]\} \in \Obs^{\rm q}(V)$ then the observable $m_{U,V} (O,O') = \{m_{U,V} (O,O') [L]\}$ is defined by
\[
m_{U,V} (O , O') [L] = \lim_{L' \to 0} W_{L < L'} \left( O [L'] \cdot O[L'] \right) .
\]
Here the $\cdot$ on the right-hand side denotes the commutative product in $\cO(\cE)[[\hbar]]$. 
To see that this is well-defined and gives $U \mapsto \Obs^{\rm q}(U)$ the structure of a (pre)factorization algebra is the content of \cite[Theorem 8.5.1.1]{Book2}. 

\begin{remark}
To handle the issue with supports one must be more careful than using an effective family $\{O[L]\}$ based simply on a ``length scale" $L$.
The correct notion is that of a {\em parametrix}, but since we will not need it here we will skip over this technical detail.
\end{remark}

We finally turn to Poisson BF theory on $\CC^{2n} \times \RR^m$.
Suppose we are in one of the cases of Theorem \ref{thm:quantization} where a translation invariant quantization is guaranteed to exist.  
Fix such a quantization and write $\obs^{\mr{q}}_{n,m}$ for the factorization algebra on $\CC^{2n} \times \RR^m$ of quantum observables.

\begin{theorem} \label{quantum_lc_thm}
The factorization algebra $\obs^{\mr{q}}_{n,m}$ of quantum observables is de Rham but not Betti topological.
\end{theorem}

\begin{proof}
First, it is straightforward to observe that $\obs^{\mr{q}}_{n,m}$ is not Betti topological.  We can recover the factorization map $\obs^{\mr{cl}}_{n,m}(B_r(0) \times U) \to \obs^{\mr{cl}}_{n,m}(B_R(0) \times U)$ in the classical factorization algebra by reducing the corresponding map in the quantum factorization algebra modulo $\hbar$.  Since this map is not an equivalence at the classical level, it cannot be an equivalence at the quantum level.

To show that $\obs^{\mr{q}}_{n,m}$ is de Rham topological, first note that it admits a smooth translation action by \cite[Proposition 10.1.1.2]{Book2}, using the fact from \ref{thm:oneloop} that our quantization is translation invariant.  We must verify that we can extend the translation action by lifting the infinitesimal $\RR^{4n+m}_{\mr{dR}}$ action to the quantum level.  To do so, we will use a result of Costello and Gwilliam on the equivariant quantization of field theories \cite[Section 12.3]{Book2}.  There is a 1-loop obstruction to lifting the $\RR^{4n+m}_{\mr{dR}}$ from the classical to the quantum level given by a 1-cocycle $\Theta^{\mr{eq}}_{n,m}$ in $C^\bullet_{\mr{red}}(\RR^{4n+m}_{\mr{dR}}, \mc O_{\mr{loc}})$.  We will check that this cocycle automatically vanishes when at least one of the inputs is a degree $-1$ element of $\RR^{4n+m}_{\mr{dR}}$.

The obstruction $\Theta^{\mr{eq}}_{n,m}$ takes a form similar to the description of the anomaly that we saw in equation \ref{holanomaly}. It is computed as a sum of weights of wheel Feynman diagrams where the external legs are labelled by fields $A, B$, or by background fields: elements of $\RR^{4n+m}_{\mr{dR}}$.  It is a straightforward observation that there are no such diagrams whose external legs are labelled by a degree $-1$ auxiliary field in $\RR^{4n+m}[1]$.  Indeed, by the definition \ref{eqn:Q} of the degree $-1$ inner symmetry $\mc Q_v$, it is purely quadratic in $v$ and $B$, and therefore degree $-1$ background fields cannot label the external legs of a wheel diagram.  As such, there is no obstruction to lifting the $\RR^{4n+m}_{\mr{dR}}$ action to the quantum level.
\end{proof}

\begin{remark}
Let $(\CC^n , \Pi)$ be a general (possibly degenerate) translation invariant holomorphic Poisson structure and consider the subspace ${\rm Im} (\Pi) \subset \RR^n$ of holomorphic translations that are in the image of $\Pi$.  
If the quantization of holomorphic BF theory on $\CC^n \times \RR^m$ associated to $\Pi$ exists, the same argument as above shows that the space of translations
\[
{\rm Im}(\Pi) \oplus \left\{\frac{\partial}{\partial \Bar{z}_i}\right\} \oplus \left\{\frac{\partial}{\partial t_j}\right\} \cong {\rm Im}(\Pi) \oplus \RR^n \oplus \RR^m
\]
act homotopically trivially, yet the holomorphic translations in $\RR^n / {\rm Im}(\Pi)$ act in a potentially non-trivial way. 
\end{remark}

\printbibliography

\end{document}